\documentclass[letterpaper, 10 pt, conference]{ieeeconf} 
\usepackage{amsmath}
\newcommand{\norm}[1]{\left\lVert#1\right\lVert}
\usepackage[export]{adjustbox}
\usepackage{verbatim}
\usepackage{comment}
\IEEEoverridecommandlockouts  

\usepackage{cite}
\usepackage{amsmath,amssymb,amsfonts}
\usepackage{algorithmic}
\usepackage{comment}
\usepackage{color}
\usepackage{graphicx}
\usepackage{epsfig} 
\graphicspath{{Figures/}}
\usepackage{textcomp}
\usepackage{xcolor}
\usepackage{flafter}
\usepackage{sidecap}
\usepackage{comment}
\usepackage{caption}
\usepackage{mathtools}
\usepackage{tikz}
\usepackage{psfrag}
\usepackage{float}
\usepackage{epstopdf}
\usepackage{subcaption}
\usepackage[ruled,vlined]{algorithm2e}
\usepackage{tabularx,booktabs}
\usepackage[ruled,vlined]{algorithm2e}
\usepackage{nomencl}
\usepackage{stackengine}
\usepackage{amssymb}
\usepackage{dsfont}
\usepackage{physics}
\newcommand{\Pc}{\mathcal{P}}
\makenomenclature
\newcommand{\triangeq}{\mathrel{\overset{\triangle}{=}}}

\usepackage{fancyhdr}
\usepackage{textcomp}
\usepackage{xcolor}
\usepackage{flafter}
\usepackage{sidecap}
\usepackage{comment}
\usepackage{caption}
\usepackage{mathtools}
\usepackage{tikz}
\usepackage{psfrag}
\usepackage{float}
\usepackage{epstopdf}
\usepackage{subcaption}
\usepackage[ruled,vlined]{algorithm2e}
\usepackage{tabularx,booktabs}
\usepackage{multirow}
\usepackage{nomencl}
\usepackage{stackengine}
\usepackage{amssymb}
\usepackage{dsfont}
\usepackage[export]{adjustbox}
\usepackage{mathtools}
\usepackage{stfloats}

\newtheorem{theorem}{\textbf{Theorem}}

\newtheorem{remark}{Remark}
\newtheorem{assumption}{\textbf{Assumption}}
\title{An Actor-Critic-Identifier Control Design for Increasing Energy Efficiency of Automated Electric Vehicles}

\author{ Hamed Faghihian$^{1}$, Arman Sargolzaei$^{1}$
\thanks{$^{1}$ Hamed Faghihian and Arman Sargolzaei are with the Department of Mechanical Engineering, University of South Florida, Tampa, FL 33620, USA. Emails: hfaghihian@usf.edu, a.sargolzaei@gmail.com}%
}

\begin{document}
\maketitle
\thispagestyle{empty}
\pagestyle{empty}

\begin{abstract}
Electric vehicles (EVs) are increasingly deployed, yet range limitations remain a key barrier. Improving energy efficiency via advanced control is therefore essential, and emerging vehicle automation offers a promising avenue. However, many existing strategies rely on indirect surrogates because linking power consumption to control inputs is difficult. We propose a neural-network (NN) identifier that learns this mapping online and couples it with an actor–critic reinforcement learning (RL) framework to generate optimal control commands. The resulting actor–critic–identifier architecture removes dependence on explicit models relating total power, recovered energy, and inputs, while maintaining accurate speed tracking and maximizing efficiency. Update laws are derived using Lyapunov stability analysis, and performance is validated in simulation. Compared to a traditional controller, the method increases total energy recovery by 12.84\%, indicating strong potential for improving EV energy efficiency.
\end{abstract}

\providecommand{\keywords}[1]
{
  \small	
  \textbf{\textit{Energy efficiency}} #1
}



\section{Introduction} \label{ssec:intro}

Despite the fact that electric vehicles (EVs) are bringing promising solutions for many of transportation problems, their market penetration remains constrained by range limitations \cite{faghihian2023energy}. However a primitive solution could be extending the battery capacity but this solution increases energy use and raises environmental, cost, and emissions concerns \cite{hung2021regionalized}. Therefore, the complementary solution is to improve efficiency through emerging capabilities in vehicles \cite{faghihian2024introduction}.

Automated vehicles (AVs) offer multiple potentials to increase energy efficiency through perception, signal-phase awareness, intersection coordination, or their combination \cite{al2018minimizing, chen2024achieving}. However, these advances are less developed for electric AVs (EAVs), where regenerative braking and different energy efficiency patterns necessitate EV-specific energy-efficient controllers development. 

Recent studies in EAV included energy optimal deceleration and braking strategies with a real time extension \cite{kim2021parameterized}, adaptive energy recovery from multi-source data \cite{shang2022regenerative}, and automated braking to extend range \cite{FAGHIHIAN20241}. However, most of these works only considered regenerative braking system, and they are not considering traffic efficiency and smoothness. Moreover, none demonstrates a controller that maximizes EAV energy efficiency within a designed driving scenario.

Energy efficiency in driving scenarios is commonly presented in the form of eco-driving which optimizes speed along a route \cite{mahmoud2021autonomous}. EV energy efficiency is also affected by battery and power electronics behavior and by mode switching, therefore consumption also depends on state dependent power flows \cite{yao2021adaptive}. As a result, fully coupled models that integrate the energy management system (EMS), vehicle dynamics, and powertrain are challenging to construct and remain underexplored.

Classical optimal control approaches including dynamic programming (DP), and model predictive control (MPC) propose powertrain models within EMS and vehicle dynamics to optimize speed regulation and energy use \cite{naeem2024energy}. These controllers performance depends highly on accurate mappings from control to power consumption or regeneration, which are frequently uncertain and sometimes inaccessible in practice \cite{jeong2024adaptive}; even when these maps are available, controller effectiveness is sensitive to model fidelity and computational power of the controller in realtime. To solve these issues, machine learning (ML) based method proposed which mitigate model-dependence. In particular, reinforcement learning (RL) has been used to learn energy-aware driving policies and enable adaptive control \cite{isele2018navigating, chen2020deep}, with applications to energy management in hybrid electric vehicles and fuel cell electric vehicles \cite{lee2021energy}. However, these studies do not target EVs and typically separates identification from control, lacking a controller that simultaneously learns the dynamics and computes an optimal policy.

Another learning based approach used for energy management in literature is Q-learning \cite{xiong2018battery, nyong2020reinforcement}. Q-learning suffers from dependency to tabular schemes which bring discretizations and weak convergence guarantees \cite{hu2019reinforcement}, Deep Q-function approximators improve scalability \cite{du2021heuristic, he2021improved}. However, they still inherit discretizations problem. Another approach to avoid discretizations introduced in literature is continuous actor–critic method  \cite{haarnoja2018soft}, this method is supervisory and can not handle unknown EV dynamics or enforce drive-cycle tracking while minimizing energy. We therefore utilize a NN-based, online approach that couples policy learning with system identification.

This paper proposes an actor–critic–identifier (ACI) controller that learns system behavior during acceleration and deceleration while tracking both system state and desired behavior allowing for generating an energy efficient control signal. Lyapunov-based analysis guides the update laws for actor, critic, and identifier. Our contributions are: (i) an ACI RL controller that improves EV energy efficiency with accurate drive-cycle tracking; (ii) elimination of explicit power-consumption modeling for adaptability across EVs and conditions; (iii)  designing actor, critic, and identifier update laws based on Lyapunov stability analysis.

The rest of this paper is organized as follows. In Section~\ref{aci.sec.2}, we present the EV model, cost, and optimal control formulation. In Section~\ref{sec:aci_method}, we present challenges, the ACI controller architecture, learning laws and update policies for each of NNs. In Section~\ref{aci.sec.5}, we establish stability via Lyapunov stability analysis. In Section~\ref{aci.sec.6}, we report simulations and compare against a tuned PID baseline. In Section~\ref{aci.sec.7}, we conclude and outline future work.

\section{Problem Formulation and Objectives} \label{aci.sec.2}
The system dynamics which serve as the foundation for constructing the optimization problem is discussed in this section.  
\subsection{Dynamic model}
The EV longitudinal unknown model can be represented as
\begin{equation} \label{aci.eq.1}
    \dot{x}(t)=F(x(t),u(t)) = g(x(t))+h(x(t))u(t),
\end{equation}
 where \( u(t) \in  U \subseteq \mathbb{R} \) is the control input, which is considered here as the torque of the EV traction motor. \( x(t) \in X \subseteq \mathbb{R}^2 \) denotes the state vector of the system with $x_1(t)\triangeq v_v(t)-v_d(t)$, where $v_v(t),~v_d(t) \in \mathbb{R}$ are the vehicle longitudinal speed and the vehicle desired speed respectively, and $x_2(t)\triangeq P_\text{acc}(t)-P_{\text{rb}}(t)$, where \( P_{\text{acc}(t)} \in \mathbb{R}\) is the power needed for acceleration, and \( P_{\text{rb}(t)} \in \mathbb{R}\) is power recovered during regenerative braking. The functions $g(x(t)),~ h(x(t)) \in \mathbb{R}^2$ in \eqref{aci.eq.1} are considered to be second-order differentiable functions.
\assumption \label{aci.ass0}
      The function \( h(x) \) considered to be known and bounded such that $\norm{h(x)} \leq \Bar{h}$, where $\Bar{h} \in \mathbb{R}^+$ is a constant value\footnote{We can consider \( h(x) \) as a bounded function since the physical model of vehicle is bonded.
}.
 
However we are assuming that $h(x)$ is known, but $g(x)$ is unknown. Therefore, the second state variable \( x_2(t) \) does not show a direct relationship to the vehicle's acceleration and the relationship is unknown\footnote{Considering the input–power mapping as unknown is physically justified in EVs as it depends on SOC, temperature, speed/torque, efficiency maps, etc., all are time-varying, explicit models are often unavailable.}, which is our main challenge addressing in this paper. For simplicity of notations, the time variable \( t \) is omitted in the remainder of the paper. 

\subsection{Problem formulation} \label{aci.subsec.PF}
Given the dynamic model \eqref{aci.eq.1}, and knowing that $g(x)$ is unknown, the problem statement is to design a control signal $u$, to minimize total energy consumption, while it keep the desired velocity. To achieve that, let's define the optimal value function ,\( J^*(x) \in \mathbb{R}\), and the subsequent control law \( u^*(x) \in \mathbb{R} \) for the system we introduced in \eqref{aci.eq.1} would be written as
\begin{equation} \label{aci.eq.2}
\begin{cases}
    J^*(x) = \min_{\substack{u(\tau) \in \Pi(X) \\ t \leq \tau < \infty}} \int_t^{\infty} \ell(x(s), u(x(s))) \, ds,~J^*(0) = 0
\\
    u^*(x) =-\tfrac{1}{2}\beta^{-1}h^T(x)(\pdv{J^*}{x}),
    \end{cases}
\end{equation}
where \( \Pi(X) \) is admissible control policies, \( \beta \in \mathbb{R^+}\) is a constant.

 To ensure speed tracking, energy efficiency, with minimum control, we designed \( \ell(x, u): X \times U \to \mathbb{R} \) as
\begin{equation} \label{aci.eq.3}
    \ell(x, u) = C(x) + u^T \beta u,
\end{equation}
where the function \( C(x) \) designed as 
\begin{equation} \label{aci.eq.3_9}
\begin{aligned}
    C(x)=&~q_1\big(v_v(t)-v_d(t)\big)^2
    +q_2 P^2_{accel}(t),
    \end{aligned}
\end{equation}
where \( q_1, q_2 \in \mathbb{R}^+ \) are weighting coefficients to penalizes deviation from the desired speed, and power drawn in acceleration phases. The acceleration terms, \( P_{accel}(t) \) defined based on $x_2(t)$ as
$P_{accel}(t)\triangeq \max(x_2(t), 0).$
This formulation allows energy consumption to be integrated into the optimization via a power parameter that can be measured experimentally or through whole-vehicle modeling.

The controller is designed by minimizing the Hamilton–Jacobi–Bellman (HJB) residual \cite{kamalapurkar2018reinforcement}. Therefor, the system Hamiltonian would be defined as
\begin{equation}\label{aci.eq.6}
\mathcal{H}(x,u,J_x) \triangleq J_x F(x,u) + \ell(x,u),
\end{equation}
where $J(x) \triangleq \int_0^\infty \ell(x(t),u(t))\,dt$, and $J_x \triangleq \tfrac{\partial J}{\partial x} \in \mathbb{R}^{1\times 2}$. At optimality we have
\begin{equation}\label{aci.eq.7}
\mathcal{H}(x,u^*,J_x^*) = J_x^* F(x,u^*) + \ell(x,u^*) = 0,
\end{equation}
which shows no further cost reduction under $u^*(x)$ (as in \eqref{aci.eq.2}). Since $u^*, J^*, F$ are unknown, we use approximations $\hat{u}, \hat{J}, \hat{F}$ respectively, leads to defining the HJB residual as
\begin{equation}\label{aci.eq.12}
\delta_{\mathrm{HJB}} \triangleq \hat{\mathcal{H}}(x,\hat{u},\hat{J}_x) - \mathcal{H}(x,u^*,J_x^*)
= \hat{J}_x \hat{F}(x,\hat{u}) + \ell(x,\hat{u}),
\end{equation}
where $\hat{J}_x \triangleq \tfrac{\partial \hat{J}}{\partial x}$. The goal is to design a continuous policy $\hat{u}$ that minimizes $\delta_{\mathrm{HJB}}$ so that the realized policy tracks the optimal one in \eqref{aci.eq.2}. 

\section{Method: Actor-Critic-Identifier (ACI)}
\label{sec:aci_method}

\subsection{Challenges}
Designing an optimal controller for the nonlinear EV system presents three major challenges.  
First, the optimal value function $J^*(x)$ and control policy $u^*(x)$ exist but are unknown and must be approximated.  
Second, the system dynamics contain an unknown drift term $g(x)$, requiring a real-time identifier to learn and compensate for these dynamics, which is the main challenge.  Third, the optimal controller must simultaneously minimize the HJB residual $\delta_{\mathrm{HJB}}$ while ensuring stability and boundedness of all estimates under time-varying driving conditions.  

\subsection{Solution Architecture}
Since we need to approximate the unknown optimal value function, policy, and dynamics, we utilize three coordinated neural networks: a \emph{critic} for approximation of $J^*(x)$, an \emph{actor} to approximate the control policy $u^*(x)$, and an \emph{identifier} to estimate the unknown part of system dynamics, $g(x)$. The three modules are coupled with the the \emph{HJB residual}, $\delta_{\mathrm{HJB}}$, therefore, in each update the residual reducing and shift the realized policy toward the optimal one. This architecture described in Figure \ref{aci.controller_scheme}.

\begin{figure}
    \centering
    \includegraphics[width=0.9\linewidth]{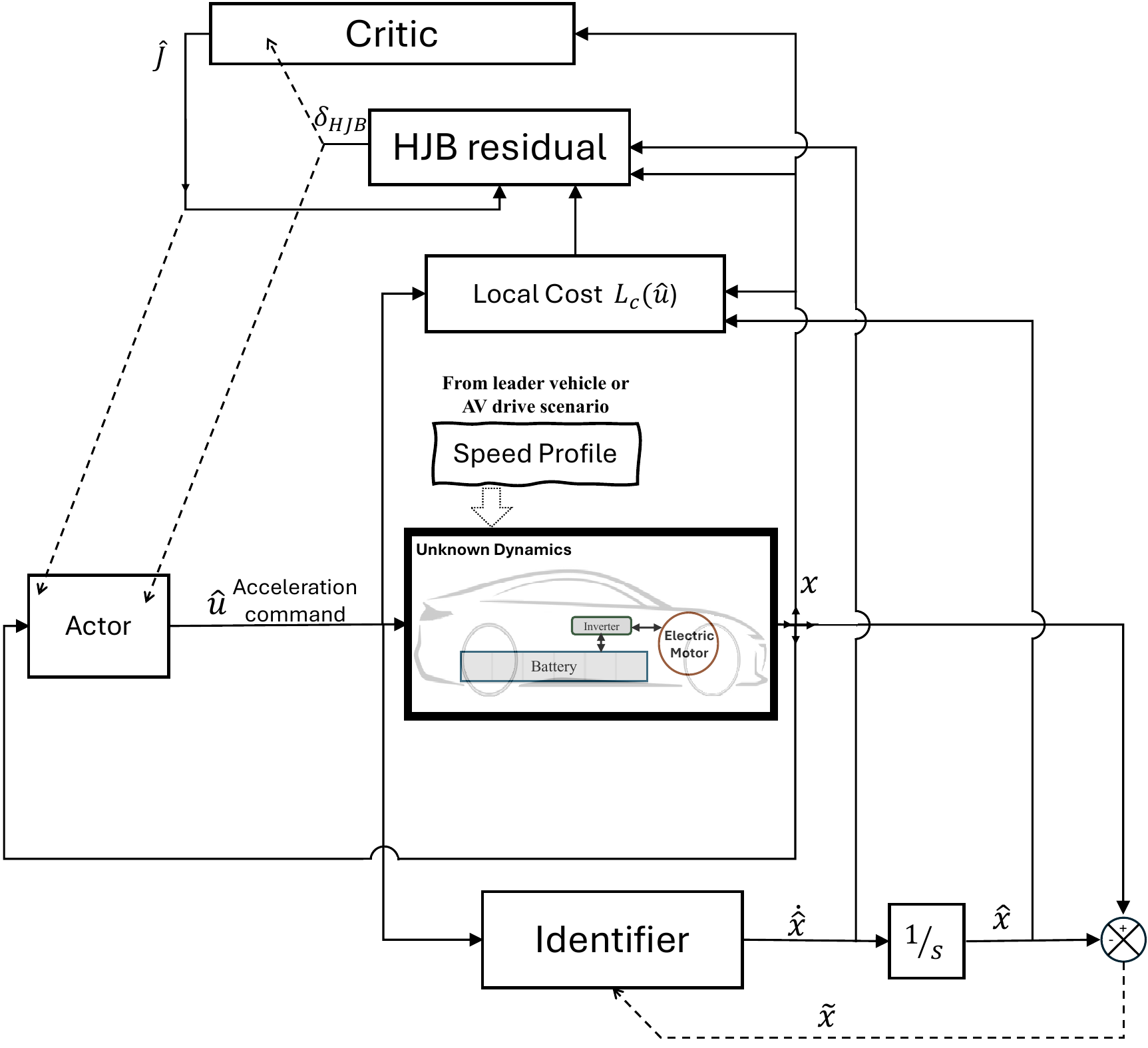}
    \caption{The controller architecture}
    \label{aci.controller_scheme}
\end{figure}

\subsection{Critic Network}
We approximate the value function as
\begin{equation} \label{aci.eq.8}
\hat J(x) = \hat w_c^{T}\,\phi(x), \qquad \phi(x) \in \mathbb{R}^N,
\end{equation}
with basis vector $\phi(\cdot)$ and critic weights $\hat w_c \in \mathbb{R}^N$. Let the regressor be defined as
\begin{equation} \label{critic_regressor}
\varphi(x, \hat{x}, \hat{u}) = \phi'(x) \hat{F}(x, \hat{u}).
\end{equation}
Based on the Lyapunov stability analysis described in Section \ref{aci.sec.5}, We design the critic update law with adaptation gain matrix, $\Pc \in \mathbb{R}^{N\cross N}$, as
\begin{equation} \label{critic_update}
\dot{\hat w}_c = -\,k_{c1}\,\frac{\Pc\,\varphi}{1+\kappa\,\varphi^{T}\Pc\varphi}\,\delta_{\mathrm{HJB}}
\;-\;k_{c2}\,(\hat w_a-\hat w_c),
\end{equation}
\begin{equation} \label{gamma_update}
    \dot{\Pc} = -\,k_{c1}\,\Pc\,\frac{\varphi\varphi^{T}}{1+\kappa\,\varphi^{T}\Pc\varphi}\,\Pc.
\end{equation}
 $\Pc $ is initializing as \( \Pc(0) = \Pc_1 I \), where \( I \in \mathbb{R}^{N \times N} \) is the identity matrix and \( \Pc_1 \in \mathbb{R}^+ \) is a positive constant with covariance resetting to keep $0\le \Pc_0 I \le \Pc(t)\le \Pc_1 I$. Therefore, based on \eqref{gamma_update}, \( \dot{\Pc}(t) \leq 0 \), showing \( \Pc(t) \) is non-increasing over time. The covariance resetting technique employed To prevent \( \Pc(t) \) to become arbitrarily small, to prevent excessively slow updates in the critic network. This would ensure stable and sufficiently responsive learning. 

\subsection{Actor Network}

The actor designed to generates the control policy \( \hat{u}(x) \) by approximating the optimal control \( u^*(x) \). Therefore, by using the approximation method of \( \hat{J}(x) \) similar to \eqref{aci.eq.8}, we can write
\begin{equation} \label{aci.eq.23.u_hat}
\hat u(x) \;=\; -\tfrac{1}{2}\,\beta^{-1}h^{T}(x)\,\phi'(x)^{T}\hat w_a,
\end{equation}
where $\hat w_a\in\mathbb{R}^N$ are actor weights. By Defining $H(x)\triangleq h(x)\beta^{-1}h^{T}(x)$. The actor update rule based on the Lyapunov stability analysis (Section \ref{aci.sec.5}) designed as
\begin{equation} \label{actor_update}
\begin{aligned}
    \dot{\hat{w}}_a = \text{proj} \Bigg\{ 
    - \frac{2 k_{a1}}{\sqrt{1 + \varphi^T \varphi}} \big[
    \left( \hat{w}_c^T \phi' \frac{\partial \hat{F}(x,\hat{u})}{\partial \hat{u}} \frac{\partial \hat{u}}{\partial \hat{w}_a} \right)^T \\~~~~~~~~~~~~~~~~~+ \frac{\partial \hat{u}}{\partial \hat{w}_a}^T \beta~\hat{u}\big] \delta_{\text{HJB}}-k_{a2}(\hat{w}_a-\hat{w}_c)
    \Bigg\},
\end{aligned}
\end{equation}
where $\mathrm{proj}\{\cdot\}$ enforcing boundedness.

\subsection{Identifier Network Design}
The identifier network provides approximations of system dynamic model. Based on the Lyapunov stability analysis in Section \ref{aci.sec.5} we designed an identifier by an NN to estimate the unknown part of the system dynamic, $g(x)$ as
\begin{equation} \label{aci.identifier2.1}
g(x) \;=\; w_g^{T}\sigma(v_g^{T}x) + \xi_g(x),
\end{equation}
where $w_g, v_g$ are ideal weights and $\xi_g$ is reconstruction error. Therefore the identifier estimates the system dynamics as
\begin{equation} \label{identifier3}
    \dot{\hat{x}}=\hat{F}(x) = \hat{w}_g^T \hat{\sigma}_g+ h(x) \hat{u}+r_t,
\end{equation}
where $\hat\sigma_g\triangleq \sigma(\hat v_g^{T}\hat x)$ and $r_t$ is a RISE feedback term defined as
\begin{equation} \label{identifier4}
\begin{aligned}
r_t \triangleq p_1\big(\tilde x(t)-\tilde x(0)\big) + \nu(t), \quad \\
\dot{\nu} = (p_1\alpha+\gamma)\tilde x + p_2\,\mathrm{sgn}(\tilde x),\ \ \nu(0)=0,
\end{aligned}
\end{equation}
where $\tilde x \triangleq x-\hat x$ and $p_1,\alpha,\gamma,p_2$ are positive gains. The identifier weight updates designed using Lyapunov stability analysis provided in Section \ref{aci.sec.5} (with constant positive $\Upsilon_w \in\mathbb{R}^{(L_g+1)\cross (L_g+1)}$ and $\Upsilon_v \in \mathbb{R}^{2 \cross 2}$) are 
\begin{equation} \label{identifier5}
\begin{aligned}
   & \dot{\hat w}_g = \mathrm{proj}\!\left\{\Upsilon_w\hat\sigma_g' \hat v_g^{T}\dot{\hat x}\tilde x^{T}\right\},
\dot{\hat v}_g = \mathrm{proj}\!\left\{\Upsilon_v\dot{\hat x}\tilde x^{T}\hat w_g^{T}\hat\sigma_g'\right\}.
\end{aligned}
\end{equation}
This ACI loop yields a continuous controller $\hat u$ that tracks the optimal policy without requiring the explicit power-consumption model of the vehicle.

\subsection{HJB Residual based on NNs}
The HJB residual with respect to \eqref{aci.eq.3}, and \eqref{aci.eq.7}-\eqref{critic_regressor} can be write as
\begin{equation} \label{aci.eq.25}
    \delta_{\text{HJB}}=\hat{w}_c^T \varphi +\hat{u}^{T} \beta \hat{u} - J^*_x F(x,u^*) - u^{*T} \beta u^*.
    \end{equation}  

Substituting \eqref{aci.eq.23.u_hat} into \eqref{aci.eq.25} gives
\begin{equation} \label{aci.eq.26}
\begin{aligned}
        \delta_{\text{HJB}}=&\hat{w}_c^T \varphi +[-\frac{1}{2}\hat{w}^T_a \phi^{\prime} h \beta^{-1}] \beta [-\frac{1}{2} \beta^{-1} h^T \phi^{\prime T} \hat{w}_a]\\& - [-\frac{1}{2}  (\xi_v^{\prime}+w^T \phi^{\prime})\beta^{-1} h] \beta [-\frac{1}{2} \beta^{-1} h^T (\phi^{\prime T}w+\xi_v^{\prime T})]\\&- [w^T \phi^\prime+\xi^\prime_v] F(x,u^*).         \end{aligned}
    \end{equation}
With $\tilde{w}_a \!\triangleq\! w-\hat{w}_a$, $\tilde{w}_c \!\triangleq\! w-\hat{w}_c$, $\tilde{u}\!\triangleq\!u^*-\hat{u}$, and $\tilde{F}(x,\hat{u})\!\triangleq\!F(x,\hat{u})-\hat{F}(x,\hat{u})$, \eqref{aci.eq.26} simplifies to
\begin{equation} \label{aci.eq.27}
\begin{aligned}
        \delta_{\text{HJB}}=&\hat{w}_c^T \varphi -w^T \phi ^\prime \Tilde{F}(x,\hat{u}) -\xi ^\prime _v F(x,u^*) \\&+\frac{1}{4}\Tilde{w}^T_a \phi^{\prime} H \phi^{\prime T} \Tilde{w}_a -\frac{1}{4}\xi_v^{\prime} H\xi_v^{\prime T}.
        \end{aligned}
    \end{equation}

\section{Stability Analysis} \label{aci.sec.5}
The actor, critic, and identifier updates are derived through Lyapunov stability analysis. To provide stability analysis we need preliminary calculations. \subsection{Necessary Assumptions and Calculations}
\assumption \label{aci.ass1}
The ACI's NN weights, $w$ and its approximations $w_a,~w_c$ are bounded such that
$w \leq \Bar{w}$, and $w_a \leq \Bar{w}_a$, $w_c \leq \Bar{w}_c$, $v_g \leq \Bar{v}_g$, and $w_g \leq \Bar{w}_g$, where $\Bar{w},~\Bar{w}_a,\text{and}~\Bar{w}_c \in \mathbb{R}^+$, $\Bar{v}_g, \Bar{w}_g \in \mathbb{R}^+$ are known constants.

\assumption \label{aci.ass2_ac}
The functions \( \phi(x), \xi_v(x) \) and their derivatives \( \phi'(x), \xi'_v(x) \) are bounded such that
\begin{equation*} \label{activation_bounds}
\begin{cases}
\|\phi(x)\| \leq M_\phi \quad \text{and} \quad \|\phi'(v_g^T x)\| \leq M_{\phi'}
\\
\|\xi_v(x)\| \leq M_{\xi_v} \quad \text{and} \quad \|\xi'_v(x)\| \leq M_{\xi'_v},
\end{cases}
\end{equation*}
where  \( M_\phi,~M_{\phi'},~M_{\xi_v},~\text{and}~ M_{\xi'_v} \) are positive definite constants
\assumption \label{aci.ass3_ac}
As $N \to \infty$ the reconstruction error and its derivatives, $\xi_v(x), \xi^\prime_v(x) \to 0$.

\assumption \label{aci.ass2}
The identifier activation $\sigma(v_g^\top x)$, its derivative, and the reconstruction error $\xi_g(x)$ with its derivative are bounded; i.e., there exist $M_\sigma, M_{\sigma'}, M_\xi, M_{\xi'}>0$ such that 
\begin{equation*} \label{activation_bounds}
\begin{cases}
\|\sigma(v_g^T x)\| \leq M_\sigma \quad \text{and} \quad \|\sigma'(v_g^T x)\| \leq M_{\sigma'}
\\
\|\xi_g(x)\| \leq M_\xi \quad \text{and} \quad \|\xi'_g(x)\| \leq M_{\xi'}.
\end{cases}
\end{equation*}

By defining \( e_g \triangleq \dot{\Tilde{x}} + \alpha \Tilde{x} \) and Using \eqref{identifier3}, \eqref{aci.identifier2.1}, and \eqref{identifier5} we can write
 \begin{equation} \label{aci.eq.proof01}
\begin{aligned}
        &\dot{e}_g=\Ddot{\Tilde{x}}+\alpha \dot{\Tilde{x}}=w_g^T\sigma_g^\prime v_g ^ T \dot{x} \\& -\big(\dot{\hat{w}}^T_g \hat{\sigma}_g + \hat{w}^T_g \hat{\sigma}_g ^\prime \dot{\hat{v}}^T_g \hat{x}+\hat{w}_g^T \hat{\sigma}^\prime_g \hat{v}^T_g\dot{\hat{x}}\big)+\dot{\xi}_g-\dot{r}_t+\alpha\dot{\Tilde{x}},
        \end{aligned}
    \end{equation}
which can be written as
\begin{equation} \label{aci.eq.proof02}
\begin{aligned}
        \dot{e}_g=T_1+T_2+T_3-p_1e_g-\gamma \tilde{x}-p_2 \text{sgn} (\tilde{x}).
     \end{aligned}
    \end{equation}    
  where  
\begin{equation} \label{aci.eq.proof03}
\begin{aligned}
& T_1 = \alpha\dot{\Tilde{x}}-\dot{\hat{w}}^T_g \hat{\sigma}_g-\hat{w}^T_g \hat{\sigma}_g^{\prime} \dot{\hat{v}}^T_g \hat{x}
+\tfrac{1}{2}w^T_g \hat{\sigma}_g^{\prime} \hat{v}^T_g \dot{\tilde{x}}
+\tfrac{1}{2}\hat{w}^T_g \hat{\sigma}_g^{\prime} v^T_g \dot{\tilde{x}},\\[2pt]
& T_2 = w^T_g \sigma_g^{\prime} v_g^T \dot{x}
-\tfrac{1}{2}w^T_g \hat{\sigma}_g^{\prime} \hat{v}^T_g \dot{x}
-\tfrac{1}{2}\hat{w}^T_g \sigma_g^{\prime} v^T_g \dot{x}
+\dot{\xi}_g,\\[2pt]
& T_3 = \tfrac{1}{2}\tilde{w}^T_g \hat{\sigma}_g^{\prime} \hat{v}^T_g \dot{\hat{x}}
+\tfrac{1}{2}\hat{w}^T_g \hat{\sigma}_g^{\prime} \tilde{v}^T_g \dot{\hat{x}}.
\end{aligned}
\end{equation}
\remark \label{aci.remark01} \label{aci.remark02} \label{aci.remark03} \label{aci.remark04} \label{aci.remark05}
Under Assumptions~\ref{aci.ass0}, \ref{aci.ass1}, \ref{aci.ass2}, the identifier \eqref{identifier3}, and the projection bounds in \eqref{identifier5}, let
$\theta \triangleq [\,\tilde{x}^\top\; e_g^\top\,]^\top$ and define
\begin{equation*}
    \begin{aligned}
&T_4 \triangleq \tfrac{1}{2}\tilde{w}_g^{\top}\hat{\sigma}_g^{\prime} \hat{v}_g^{\top}\dot{\tilde{x}}
           + \tfrac{1}{2}\hat{w}_g^{\top}\hat{\sigma}_g^{\prime} \tilde{v}_g^{\top}\dot{\tilde{x}},
\\
&T_5 \triangleq T_3 + T_4
           = \tfrac{1}{2}\tilde{w}_g^{\top}\hat{\sigma}_g^{\prime} \hat{v}_g^{\top}\dot{x}
           + \tfrac{1}{2}\hat{w}_g^{\top}\hat{\sigma}_g^{\prime} \tilde{v}_g^{\top}\dot{x},\\
&T_6 \triangleq T_2 + T_5.
    \end{aligned}
\end{equation*}
Then there exist nondecreasing functions $\rho_1,\rho_2:\mathbb{R}^+\!\to\!\mathbb{R}^+$ and constants $c_1,c_2,c_3,c_4,c_5,c_6\in \mathbb{R}^+$ such that
\begin{equation*}
    \begin{aligned}
&\|T_1\| \le \rho_1(\|\theta\|)\,\|\theta\|,~~~~~~~~~~~~~~~~~~~~~~
\|T_2\| \le c_1,\\
&\|\dot{\tilde{x}}^{\top}T_4\| \le c_2\|\dot{\tilde{x}}\|^2 + c_3\|e_g\|^2,~~~~~~~~~~~~
\|T_5\| \le c_4,\\ &
\|T_6\| \le c_1 + c_4,~~~~~~~~~~~~~~
\|\dot{T}_6\| \le c_5 + c_6\,\rho_2(\|\theta\|)\,\|\theta\|.
    \end{aligned}
\end{equation*}

\subsection{Proof of Stability}
\begin{theorem} \label{aci.th.1}

With the update laws \eqref{identifier5}, the identifier \eqref{identifier3} guarantees asymptotic tracking $\|\tilde{x}(t)\|,\,\|\dot{\tilde{x}}(t)\|\!\to\!0$ as $t\!\to\!\infty$ whenever
\begin{equation}
\begin{aligned}
&p_1>2c_3 ~~~~~~ p_2 > \text{max} \{c_1+c_4,+c_1+\tfrac{c_5}{\alpha}\},\\& p_3>c_6,~~~~~~ \gamma> \tfrac{2c_2}{\alpha}.
\end{aligned}
\end{equation}

    \end{theorem}
    \begin{proof}Define \(y_1 \triangleq [\,\tilde{x}^\top\; e_g^\top\; \sqrt{\Theta}\; \sqrt{\Omega}\,]^\top \in \mathbb{D}\subset\mathbb{R}^6\) and the Lyapunov candidate function as
\begin{equation}\label{aci.eq.proof1.1}
V_{L_1}=\tfrac12 e_g^\top e_g+\tfrac12\gamma\,\tilde{x}^\top\tilde{x}+\Theta+\Omega, ~~~~~~V_{L_1} : \mathbb{D} \to \mathbb{R}
\end{equation}
where 

\[\Omega \triangleq \tfrac{\alpha}{4}\!\left(\operatorname{tr}\!\left(\tilde{w}_g^\top\Upsilon_w^{-1}\tilde{w}_g\right)+\operatorname{tr}\!\left(\tilde{v}_g^\top\Upsilon_v^{-1}\tilde{v}_g\right)\right).\]
   
\remark \label{aci.remark07}
Let $\Theta(\theta,t)\in\mathbb{R}$ be the Filippov solution of \eqref{aci.eq.proof1.2} \cite{filippov2013differential}:
\begin{equation}\label{aci.eq.proof1.2}
\begin{aligned}
\dot{\Theta}&=- e_g^T \left( T_2 - p_2\, \text{sgn}(\tilde{x}) \right) 
    - \dot{\tilde{x}}^T T_5 
    + p_3 \rho_2(\|\theta\|) \|\theta\| \|\tilde{x}\|,
    \\&\Theta(0) = p_2 \sum_{i=1}^n \left| \tilde{x}_i(0) \right| - \tilde{x}^T(0) T_6(0),
\end{aligned}
\end{equation}
Then $\Theta(t)\ge 0$ for all $t\ge 0$ if the gains satisfy
$p_2>\max{,c_1+c_4,; c_1+c_5/\alpha,}$ and $p_3>c_6$.   
    
    Applying conditions for Remark \ref{aci.remark07} with \eqref{aci.eq.proof1.1} we can find positive definite bonds, $b_{11}, b_{12}$ on $V_{L_1}$, such that
\begin{equation}\label{aci.eq.proof1.3}
    b_{11}(y_1) \leq V_{L_1}(y_1) \leq b_{12}(y_1),
\end{equation}
where \( b_{11}(y) \triangleq \frac{1}{2} \min(1, \gamma)\) and \( b_{12}(y) \triangleq \max(1, \gamma) \|y_1\|^2\).

Using Filippov’s framework, a generalized Lyapunov analysis utilized. Therefore, generalized time derivative of \eqref{aci.eq.proof1.1} for almost all \(t\in[t_0,t_f]\), can be written as
\begin{equation}\label{aci.eq.proof1.4}
\dot V_{L_1}(y_1)\ \in\ \bigcap_{\chi\in\partial V_{L_1}(y_1)} \chi^\top\,K[Y](s,t),
\end{equation}
where $Y(y_1,t)\ \triangleq\ \begin{bmatrix}
\dot e_g^\top & \dot{\tilde{x}}^\top & \tfrac{1}{2}\Theta^{-\frac{1}{2}}\dot\Theta & \tfrac{1}{2}\Omega^{-\frac{1}{2}}\dot\Omega
\end{bmatrix}^\top,
$ and \(K[\cdot]\) is set-valued map defined as $K[J](s, t) \triangeq \bigcap_{\delta > 0} \bigcap_{R_N = 0} \Bar{co} ~J(B(s, r) - N, t), $ where \(\Bar{co}\) is convex closure with \( B(s,r) \) is a ball of radius \( r \) centered at \( s \). 

Since \(V_{L_1}(y_1)\) is Lipschitz and regular, Filippov’s operator with using the generalized Lyapunov method for discontinuous systems \cite{guo2009generalized}, the right hand side of \eqref{aci.eq.proof1.4} can be written as
\begin{equation} \label{aci.eq.proof1.5}
    \big[{e}_g^T~ \gamma\dot{\tilde{x}}^T ~2 \Theta^{-\frac{1}{2}} ~ 2 \Omega^{-\frac{1}{2}}  \big]~  K ~\big[\dot{e}_g^T~ \dot{\tilde{x}}^T ~\frac{1}{2} \Theta^{-\frac{1}{2}} \dot{\Theta} ~ \frac{1}{2} \Omega^{-\frac{1}{2}} \dot{\Omega} \big]^T,
\end{equation}
 which we define the \eqref{aci.eq.proof1.5} as $V_R$, so we can consider \(\dot V_{L_1}(y_1)\in V_{R}\). Applying \eqref{aci.eq.proof02}, \eqref{aci.eq.proof1.2}, into \eqref{aci.eq.proof1.5}, we can rewrite $V_R$ as
 \begin{equation} \label{aci.eq.proof1.6}
 \begin{aligned}
       V_R=&e_g^T(T_1+T_2+T_3-p_1e_g-\gamma \tilde{x}-p_2 K[\text{sgn} (\tilde{x})]) 
        \\&+\gamma \tilde{x}(e_g-\alpha\tilde{x}) - e_g^T \left( T_2 - p_2\, K[\text{sgn}(\tilde{x})] \right)\\& 
    - \dot{\tilde{x}}^T T_5 
    + p_3 \rho_2(\|\theta\|) \|\theta\| \|\tilde{x}\|\\&
   - \frac{1}{2} \alpha 
    \left[ 
        \text{tr}\left( \tilde{w}_g^T \Upsilon_{w}^{-1} \dot{\hat{w}}_g \right) 
        + \text{tr}\left( \tilde{v}_g^T \Upsilon_{v}^{-1} \dot{\hat{{v}}}_g \right)
    \right],
     \end{aligned}
\end{equation}
where  \(K[\operatorname{sgn}(\tilde{x})]=\{1\}\) for \(\tilde{x}>0\), \([-1,1]\) for \(\tilde{x}=0\), and \(\{-1\}\) for \(\tilde{x}<0\). Based on this, \eqref{aci.eq.proof1.6} will be reduced
to scalar, also by using \eqref{identifier5}, \eqref{aci.eq.proof03}, and with considering Remark \ref{aci.remark04}, we can rewrite and simplify it to 
\begin{equation}\label{aci.eq.proof1.7}
    T_1 e_g-p_1 e^T_g  e_g-\alpha\gamma\tilde{x}^T\tilde{x}-T_4 \dot{\tilde{x}}^T+p_3 \rho_2(\|\theta\|) \|\theta\| \|\tilde{x}\|.
\end{equation}
Since the expression in \eqref{aci.eq.proof1.7} is continuous for \(\tilde{x}(t)\neq 0\) and any discontinuities occur on a set of times of Lebesgue measure zero, by Remarks~\ref{aci.remark01} and~\ref{aci.remark03} we can write 
\begin{equation}\label{aci.eq.proof1.8_1}
\begin{aligned}
V_{R} \leq& -\big(\frac{p_1}{2}-c_3\big) ||e_g||^2 -\big(\frac{\alpha\gamma}{2}-c_2\big)||\tilde{x}||^2\\& +
 \big(\frac{1}{2p_1}\rho_1(||\theta||)^2\big)+\big(\frac{1}{2\alpha\gamma}p_2^2 [\rho_2(\|\theta\|)]^2 \big)||\theta||^2,
 \end{aligned}
\end{equation}
which can be written as
\begin{equation}\label{aci.eq.proof1.9}
\begin{aligned}
V_R \leq& -l_1 ||\theta||^2+l_2((||\theta||)^2||\theta||^2,
\end{aligned}
\end{equation}
where $l_1\triangeq min\{\frac{p_1}{2}-c_3, \frac{\alpha\gamma}{2}-c_2\}$ and $l_2(\theta) \triangeq max \{ (\frac{1}{2\alpha\gamma} p_2^2 [\rho_2(\|\theta\|)]^2 ,\frac{1}{2p_1}[\rho_1(||\theta||)]^2  \}$. Then the right hand side of \eqref{aci.eq.proof1.8_1} can be bounded as

\begin{equation}\label{aci.eq.proof1.9}
\begin{aligned}
V_R \leq& -l_1 ||\theta||^2+l_2((||\theta||)^2||\theta||^2,
\end{aligned}
\end{equation}
therefore, there exist \(p_1\) and \(\gamma\) such that $ V_R\leq -l_3||\theta||^2$
where $l_3$ is a positive constant. This will leads to
\begin{equation}\begin{aligned}\label{aci.Stability_9}
\dot{{V}}_{L_1}(y_1)\leq-\Gamma_1(y_1),~~~~~~~\forall \dot{V}_{L_1}(y_1) \mathrel{\overset{\text{a.e.}}{\in}} \dot{R}_{L_1}(y_1),
\end{aligned}
\end{equation}
where $\Gamma_1(y_1)\triangeq l_3||\theta||^2$ is a positive, continuous, semi-definite function. applying bonds in \eqref{aci.eq.proof1.3}, we can write
\begin{equation}\begin{aligned}\label{aci.Stability_11}
\dot{{V}}_{L_1}(y_1)\leq-\frac{\Gamma_1(y_1)}{ b_{12}(y_1)} V_{L_1}.
\end{aligned}
\end{equation}
Therefore, with choosing positive gains \( \gamma > \tfrac{c_2}{\alpha} \) and \( p_1 > c_3 \) ensures, via \eqref{aci.eq.proof1.1}, asymptotic tracking such that
\[
\lim_{t\to\infty}\|\tilde{x}(t)\|=0,\qquad
\lim_{t\to\infty}\|\dot{\tilde{x}}(t)\|=0.
\]
\end{proof}
\begin{theorem} \label{aci.th.2}
Holding conditions of Theorem \ref{aci.th.1} with the identifier \eqref{identifier3} and update laws \eqref{identifier5}, $x(t)$, $\tilde{w}_a$, and $\tilde{w}_c$ are uniformly ultimately bounded (UUB) if and only if
\begin{equation} \label{aci.eq.proof2.0}
     \begin{cases}
         k_{a2} > \frac{k_{a1}}{4}c_8c_9+\frac{k_{a2}\epsilon_1}{2}+\frac{k_{c2}}{2\epsilon_2},
         \\
       k_{c2}+\frac{k_{c2}\epsilon_2}{2}+\frac{k_{a2}}{2\epsilon_1} >    
       \frac{\epsilon_3}{2}+k_{a1}c_7c_8+k_{c1}\frac{\Pc_1}{\kappa\Pc_0}.
       
     \end{cases}
 \end{equation}
\end{theorem}
\begin{proof}
 With the Lyapunov candidate function
\begin{equation} \label{aci.eq.proof2.1}
   V_{L_2}(x, \tilde{w}_c, \tilde{w}_a) =J^*(x)+\frac{1}{2}\tilde{w}_c^T \tilde{w}_c+\frac{1}{2}\tilde{w}_a^T \tilde{w}_a,
\end{equation}
and by introducing 
\begin{equation*}
\begin{aligned}
        b_{21}(||y_2||) &\triangeq \rho_{J1}+\tfrac{1}{2}||\tilde{w}_c||^2 +\tfrac{1}{2}||\tilde{w}_a||^2, \\
    b_{22}(||y_2||) &\triangeq \rho_{J2}++\tfrac{1}{2}||\tilde{w}_c||^2 +\tfrac{1}{2}||\tilde{w}_a||^2.
\end{aligned}
\end{equation*} where $y_2 \triangeq [x^T, \tilde{w}_c^T, \tilde{w}_a^T]^T$, and we can show there exist \(\rho_{J1}(\|x\|)\) and \(\rho_{J2}(\|x\|): [0, +\infty) \to [0, +\infty)\), as the lower and upper bounds of $J^*$. resulting in  \(b_{21}(\|y_2\|)\) and \(b_{22}(\|y_2\|)\) to serve as upper and lower bonds of $V_{L_2}$.

With considering the control input $\hat{u}$, Using \eqref{aci.eq.1}, \eqref{aci.eq.2} and  taking time derivative of $J^*$, with respect to \eqref{aci.eq.12}
\begin{equation} \label{aci.eq.proof2.2_2}
\begin{aligned}
\dot{J}^*=&-\pdv{J^*}{x}gu^*-C(x)-u^{*T} \beta u^*+\pdv{J^*}{x}g\hat{u}.
    \end{aligned}
\end{equation}
Therefore, the time derivative of the Lyapunov candidate function with respect to \eqref{critic_update},and \eqref{actor_update} would become

\begin{equation} \label{aci.eq.proof2.2}
\begin{aligned}
\dot{V}_{L_2}=&-\pdv{J^*}{x}gu^*-\Omega(x)-u^{*T} \beta u^*-2u^{*T}\beta
        \\&
   +\tilde{w}_c^T[k_{c1} \Pc \frac{\varphi}{1+\kappa\varphi^T\Pc\varphi}\delta_{HJB}+k_{c2}(\hat{w}_a-\hat{w}_c)]  
   \\&
   +\tilde{w}_a^T [\frac{k_{a1}}{\sqrt{1+\varphi^T\varphi}}\phi^{\prime}H\phi^{\prime T} (\hat{w}_a-\hat{w}_c)\delta_{HJB}
   \\&
    +k_{a2}(\hat{w}_a-\hat{w}_c)].
    \end{aligned}
\end{equation}

\vspace{1pt}Using \eqref{aci.eq.2}, \eqref{aci.eq.6} and, as $t\to\infty$, $\hat w_a-\hat w_c=\tilde w_a-\tilde w_c$, and by defining $\Psi_1 \triangleq \varphi/\sqrt{1+\kappa\,\varphi^{\top}\Pc\varphi}$ to rewrite \eqref{aci.eq.proof2.2}; then, invoking \eqref{aci.eq.27} and introducing $\Psi_2 \triangleq \frac{k_{a1}}{\sqrt{1+\varphi^{\top}\varphi}}\;\tilde w_a^{\top}\phi' H \phi'^{\top}$, the $\dot{V}_{L_2}$ upper bounded as
\begin{equation} \label{aci.eq.proof2.5} 
\begin{aligned}
&\dot{V}_{L_2}\leq-C(x)
   +\tfrac{1}{2}w^T \phi ^\prime H \xi^{\prime T}_v+\tfrac{1}{2}\xi^{\prime}_v H \xi^{\prime T}_v+\tfrac{1}{2}w^T \phi^{\prime} H \phi^{\prime T} \tilde{w}_a
 \\&
   +\frac{1}{2}\xi^{\prime}_v H \phi^{\prime T}_v \tilde{w}_a-\Psi_2 \tilde{w}_a\tilde{w}_c^T \varphi
   -\Psi_2 \tilde{w}_a w^T \phi ^\prime \tilde{F}(x,\hat{u})
   \\&
   +\Psi_2 \tilde{w}_c \tilde{w}^T_c \varphi
      +\tfrac{1}{4}\Psi_2 \tilde{w}_a \tilde{w}^T_a  \phi^{\prime} H \phi^{\prime T}   
      -\tfrac{1}{4} \Psi_2 \tilde{w}_a \xi^{\prime}_v H \xi^{\prime T}_v 
      \\&
+ \Psi_2 \tilde{w}_c w^T \phi ^\prime \tilde{F}(x,\hat{u})
    - \Psi_2 \tilde{w}^T_a\phi^{\prime} H \phi^{\prime T} \tilde{w}_a \xi^\prime_v F(x,u^*)
            \\&
      -\tfrac{1}{4}\Psi_2 \tilde{w}_c \tilde{w}^T_a \phi^\prime  H \phi^{\prime T}
      +\tfrac{1}{4}\Psi_2  \tilde{w}_c \xi^{\prime}_v H \xi^{\prime T}
      \\&
        +\Psi_2 \tilde{w}_c  \xi^{\prime}_v H \xi^{\prime T}F(x,u^*)
    +k_{c1} \tilde{w}_c ^T \Pc \tfrac{\Psi_1 \Psi_1^T}{\varphi} \xi^\prime_v H \xi^{\prime T}_v 
\\&
+k_{c1} \tilde{w}_c ^T \Pc \tfrac{\Psi_1 \Psi_1^T}{4\varphi}w_a^T\phi^\prime H \phi ^{\prime T}
-k_{c1} \tilde{w}_c ^T \Pc \tfrac{\Psi_1 \Psi_1^T}{4\varphi}\xi_v ^\prime H \xi_v ^{\prime T}
\\&
-k_{c1} \tilde{w}_c ^T \Pc \tfrac{\Psi_1 \Psi_1^T}{\varphi}w^T\phi^\prime \tilde{F}(x,\hat{u})
-k_{c1} \tilde{w}_c ^T \Pc \tfrac{\Psi_1 \Psi_1^T}{\varphi}\xi_v ^\prime F(x,u^*)
\\&
+k_{a2} \tilde{w}_a ^T \tilde{w}_c
-k_{a2} \tilde{w}_a ^T \tilde{w}_a
+k_{c2} \tilde{w}_c ^T \tilde{w}_a
-k_{c2} \tilde{w}_c ^T \tilde{w}_c.
    \end{aligned}
\end{equation}

\remark \label{aci.rem.3} Under Assumptions \ref{aci.ass0}-\ref{aci.ass2_ac}, there exists $c_7,~c_8,$$~c_{10},~c_9\in\mathbb{R}^+$ such that
\begin{equation*}
    \begin{aligned}
&||\tilde{w}_a|| \leq c_7,~~~~~~\|\phi' H \phi'^{\top}\|\le c_8,~~~~~~\|\xi_v' H \xi_v'^{\top}\|\le c_9,
\\&
\Big\|\tfrac{1}{2}w^{\top}\phi' H \xi_v'^{\top}
   + \tfrac{1}{2}\xi_v' H \xi_v'^{\top}
   + \tfrac{1}{2}w^{\top}\phi' H \phi'^{\top}\tilde w_a\\&~~~~~~~~~~~~~~~~~~~~~~~~~~~~~~~~~~~~~~
   + \tfrac{1}{2}\xi_v' H \phi'^{\top}\tilde w_a\Big\|\le c_{10}.
\end{aligned}
\end{equation*}
  \remark \label{aci.rem.c5c6} Holding Theorem \ref{aci.th.1} there exist $c_{11}, c_{12} \in \mathbb{R}^+$ such that $||w^T \phi^\prime \tilde{F}(x,\hat{u})||\leq c_{11},~~ ||\xi^\prime_v F(x,u^*)||\leq c_{12}$

Using \eqref{aci.eq.proof2.5} together with Remarks \ref{aci.rem.3}, \ref{aci.rem.c5c6}, the bound $|\Psi_1\Psi_1^T|\le 1/(\kappa\Pc_0)$, and Young’s inequality with $\epsilon_1, \text{and} ~\epsilon_2 \in \mathbb{R}$ as positive constants, we obtain

\begin{equation} \label{aci.eq.proof2.9}
\begin{aligned}
&\dot{V}_{L_2}\leq-C(x)+c_{10}
-(k_{a2}-\frac{k_{a1}}{4}c_8c_9+\frac{k_{a2}\epsilon_1}{2}+\frac{k_{c2}}{2\epsilon_2})||\tilde{w}_a||^2
\\&
-(k_{c2}-k_{a1}c_7c_8-k_{c1}\frac{\Pc_1}{\kappa\Pc_0}+\frac{k_{a2}}{2\epsilon_1}+\frac{k_{c2}\epsilon_2}{2})||\tilde{w}_c||^2
\\&+\big[k_{a1}c_7c_8(c_7+c_{11}+c_7c_8+c_9+c_{12})+\\&~~~~~~~~~~~~~~k_{c1}\frac{\Pc_1}{2\sqrt{\kappa\Pc_0}}(c_{11}+c_7c_8+\frac{c_9}{4}+c_{12})\big]||\tilde{w}_c||\\&
+k_{a1}c^2_7c_8c_{11}+c_7^2c_8+c_7c_{12}.
    \end{aligned}
\end{equation}
With more simplification, and applying Young’s inequality Applying Young’s inequality along with bonds in\eqref{aci.eq.proof2.0} we can write
\begin{equation} \label{aci.eq.proof2.11}
\begin{aligned}
&\dot{V}_{L_2}\leq-C(x)
-\big[(k_{a2}-\frac{k_{a1}}{4}c_8c_9+\frac{k_{a2}\epsilon_1}{2}+\frac{k_{c2}}{2\epsilon_2})\big]||\tilde{w}_a||^2
\\&
-\Big[(k_{c2}-k_{a1}c_7c_8-k_{c1}\frac{\Pc_1}{\kappa\Pc_0}+\frac{k_{a2}}{2\epsilon_1}+\frac{k_{c2}\epsilon_2}{2})-\frac{\epsilon_3}{2}\Big]||\tilde{w}_c||^2
\\&+\frac{1}{2\epsilon_3} \Big[k_{a1}c_7c_8(c_7+c_{11}+c_7c_8+c_9+c_{12})+\\&~~~~~~~~~~~~~~~~~~~~k_{c1}\frac{\Pc_1}{2\sqrt{\kappa\Pc_0}}(c_{11}+c_7c_8+\frac{c_9}{4}+c_{12})\Big]^2
\\&
+(k_{a1}c_7c_8)c_7c_{11}+c_7^2c_8+c_7c_{12}+c_{10}.
    \end{aligned}
\end{equation}
where $\epsilon_3 \in \mathbb{R}$ is a positive constant. 
Since \(C(x)\) is positive definite, there exists a strictly increasing function \(\Gamma_2:[0,\infty)\!\to\![0,\infty)\) such that
\begin{equation*}
    \begin{aligned}
        &\Gamma_2(\|y_2\|)\;\le\; C(x)
+\Big(k_{a2}-\tfrac{k_{a1}}{4}c_8c_9+\tfrac{k_{a2}\epsilon_1}{2}+\tfrac{k_{c2}}{2\epsilon_2}\Big)\|\tilde w_a\|^2\\&
~~~~~+\Big(k_{c2}-k_{a1}c_7c_8-k_{c1}\tfrac{\Pc_1}{\kappa\Pc_0}+\tfrac{k_{a2}}{2\epsilon_1}+\tfrac{k_{c2}\epsilon_2}{2}-\tfrac{\epsilon_3}{2}\Big)\|\tilde w_c\|^2.
    \end{aligned}
\end{equation*}
Therefore, from \eqref{aci.eq.proof2.11}, we can write,
\begin{equation} \label{aci.eq.proof2.12}
\begin{aligned}
&\dot{V}_{L_2}\leq-\Gamma_2(||y_2||)+k_{a1}c^2_7c_8c_{11}+c_7^2c_8+c_7c_{12}+c_{10}
\\&+\frac{1}{2\epsilon_3} \Big[k_{a1}c_7c_8(c_7+c_{11}+c_7c_8+c_9+c_{12})+\\&~~~~~~~~~~~~~~~~~~~k_{c1}\frac{\Pc_1}{2\sqrt{\kappa\Pc_0}}(c_{11}+c_7c_8+\frac{c_9}{4}+c_{12})\Big]^2.
    \end{aligned}
\end{equation}
With considering bonds on the Lyapunov candidate function, we can write
\begin{equation} \label{aci.eq.proof2.13}
\begin{aligned}
&\dot{V}_{L_2}\leq-\frac{\Gamma_2(||y_2||)}{b_{22}(||y_2||)}V_{L_2}+k_{a1}c^2_7c_8c_{11}+c_7^2c_8+c_7c_{12}+c_{10}
\\&+\frac{1}{2\epsilon_3} \big[k_{a1}c_7c_8(c_7+c_{11}+c_7c_8+c_9+c_{12})+\\&~~~~~~~~~~~~~~~~k_{c1}\frac{\Pc_1}{2\sqrt{\kappa\Pc_0}}(c_{11}+c_7c_8+\frac{c_9}{4}+c_{12})\big]^2.
    \end{aligned}
\end{equation}
The inequality in \eqref{aci.eq.proof2.13} ensure semi-global UUB of \(y_2\). Moreover, by Assumption~\ref{aci.ass3_ac}, increasing the NN size will yields \(\xi_v' \to 0\); therefore, as \(t\!\to\!\infty\), the actor-critic weight errors satisfy \(\tilde w_a \to 0\) and \(\tilde w_c \to 0\).
\end{proof}

\section{Results} \label{aci.sec.6}
To evaluate the ACI controller on an EV with unknown dynamics, we embed a realistic model of powertrain–dynamics plant from \cite{faghihian2025novel} into the loop as described in Figure.~\ref{aci.controller_scheme} and the plant is treated as a black box. The controller identify the unknown system dynamics online, track the desired speed, and improve energy efficiency. We assume the function \(h(x)\) is known and constant, chosen as $h_1\big((v_v-v_d),P_{ntp}\big)=0,~h_2\big((v_v-v_d),P_{ntp}\big)=1$. The NN structures and parameter values are summarized in Table~\ref{aci.tab.paras}.

\begin{table}[] 
\centering
\caption{NN structures and controller gains}
\begin{tabular}{lll} 
\toprule
Block & Item & Definition / Value \\
\midrule
\multirow{6}{*}{Identifier}
&Activation,~$L_g$& $\displaystyle \sigma_g(z)=\frac{2}{1+\exp(-v_g^{T}\hat x)}-1\;,~5 $ \\
&$p_1,$~~$p_2$ & $80,~~0.2$ \\
& $\alpha,~~~\gamma$ & $300,~~5$ \\
& $\Upsilon_{w},\,\Upsilon_{v}$ & $0.1\,I_{6\times6},~\;0.1\,I_{2\times2}$ \\
\midrule
\multirow{5}{*}{Actor--Critic}
& $\phi(x)$,~~$N$ & $\displaystyle \phi(x)=\begin{bmatrix} x_1^2 & x_1x_2 & x_2^2 \end{bmatrix}^{T}$,~~3 \\
& $k_{a1},~~k_{a2}$ & $10,~~50$ \\
& $k_{c1},~~k_{c2}$ & $11,~~30$ \\
& $\kappa$ & $0.005$ \\
\bottomrule
\end{tabular}\label{aci.tab.paras}
\end{table}

To evaluate the proposed controller, we considered a representative drivecycle that can be generated by a motion planner system for AVs or induced by traffic. The EV supposed to track the drivecycle profile. As a baseline, we implemented and tuned a PID controller for the same objective, and we compared performance of our proposed controller against the PID controller performance. This comparison shown in Figure \ref{aci.fig.drivCycleComparison}.

\begin{figure} [] 
    \centering
    \includegraphics[width=0.85\linewidth]{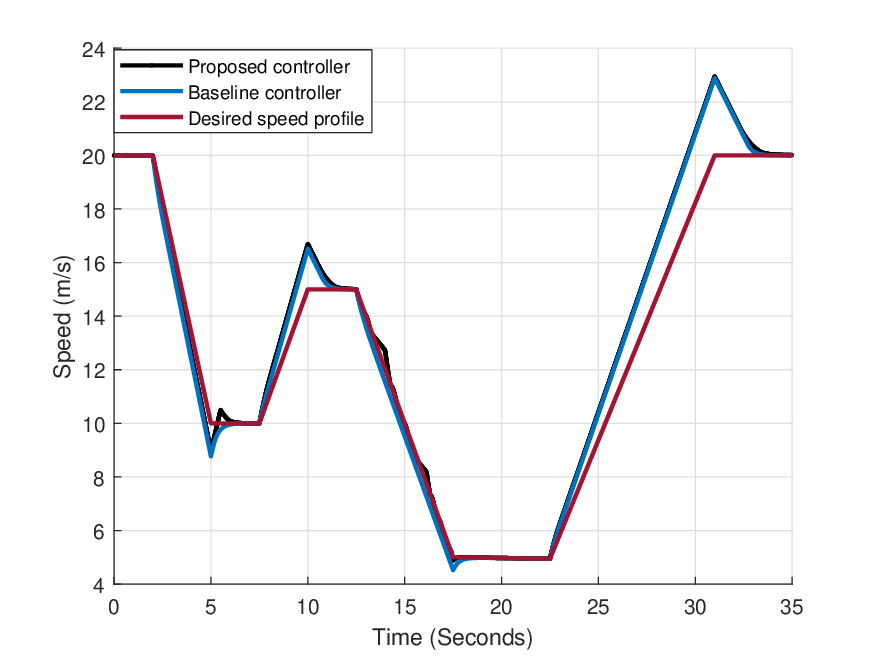}
    \caption{Desired profile versus trajectories under the proposed ACI controller and the baseline controller.}
    \label{aci.fig.drivCycleComparison}
\end{figure}

  \begin{figure}[h!] 
  \centering
    \includegraphics[width=0.85\linewidth]{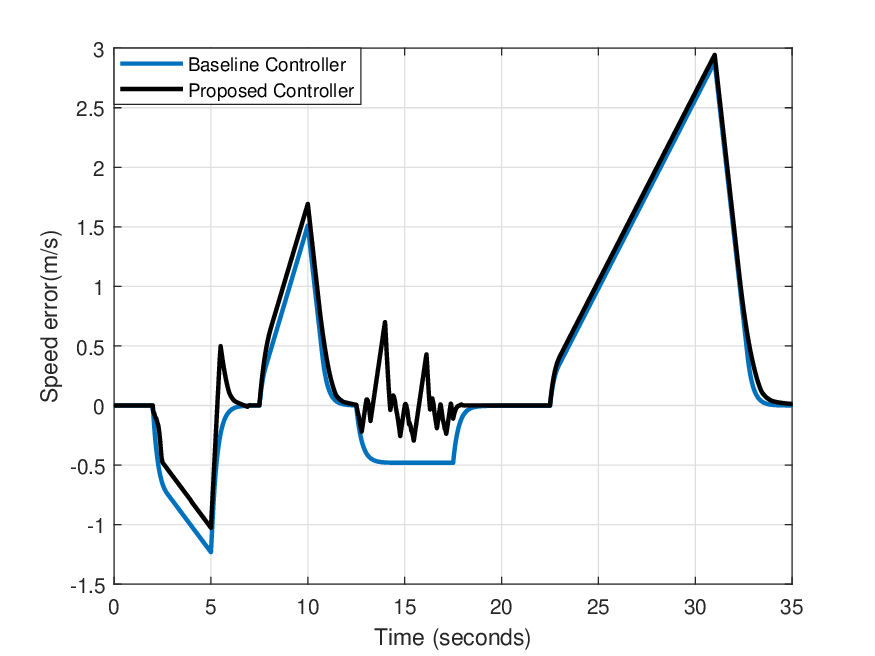}
    \caption{Speed-tracking error over the drivecycle for the proposed ACI controller and the baseline controller.}
        \label{aci.fig.statesComp}
\end{figure}

\begin{figure}[h]
    \centering
    \includegraphics[width=0.85\linewidth]{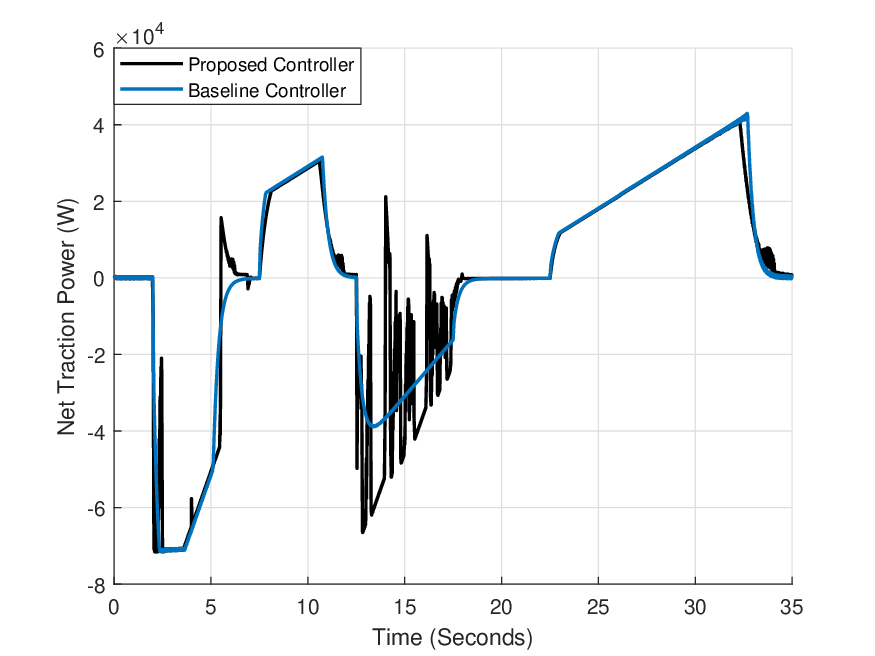}
\caption{Net traction power over the drive cycle for the proposed ACI controller and the baseline controller.}    
    \label{aci.fig.powerComp}
\end{figure}

Since the controller is designed to drive the system states to zero, we showed the controller states trajectories in Figure \ref{aci.fig.statesComp} and \ref{aci.fig.powerComp} respectively.
Also in order o be able to measure how much the total net traction power\footnote{The net traction power defined as the power used for longitudinal acceleration/deceleration; aerodynamic, rolling-resistance, and friction losses are excluded.} The Figure \ref{aci.fig.powerComp} also shows each controller power took from the EV battery. Note that the positive values of power means that the power used from battery and the negative power showing the power recharging back to the EV battery due to the RBS.

The net traction energy over the drive cycle is obtained by integrating the net traction power signal. We can perform this measurement to indicate how much energy consumed or recaptured under same drivecycle with different controllers. The baseline PID consumes almost \(13.835\ \mathrm{kJ}\), whereas the proposed ACI controller consumes almost \(7.99\ \mathrm{kJ}\), which is a \(42\%\) reduction. Because the initial and final speeds are equal, this gap primarily reflects drivetrain conversion losses, which the proposed ACI controller reduces those losses. The comparison of this effect shown in Table \ref{aci.tab.2}. 

\begin{table}[h]\centering
\caption{Comparison of  total energy consumption between the proposed and baseline controllers}\label{aci.tab.2}
\begin{tabular}{|lc|}
\hline
\textbf{Controller}               &  \textbf{Energy  consumed}  \\ \hline 
Proposed ACI controller        & 7.99 KJ                          \\
Baseline PID controller              & 13.835 KJ
\\
\hline
\end{tabular}
\end{table}

\section{Conclusion} \label{aci.sec.7}
This study showed a design procedure for an ACI controller aiming at maximizing energy efficiency while maintaining accurate drive-cycle tracking for an EV with unknown dynamics. For the design purpose a Lyapunov/Filippov analysis utilized.

Simulations on a realistic integrated EV powertrain–dynamics model show that the proposed controller achieves lower tracking error than a tuned PID baseline and reduces net traction energy over the drivecycle by a \(42\%\). These results indicate that energy-aware optimal control can be realized without an explicit power-consumption model, relying instead on online identification and policy learning.

Future works will focus on hardware-in-the-loop and on-road experiments, and integration with connected-vehicle information. Another possible future work would be developing a two-mode controller design that treats coasting/regenerative braking separately from traction/acceleration, we expect this separation would increase energy recovery and overall efficiency.

\bibliography{Ref.bib}

@book{kamalapurkar2018reinforcement,
  title={Reinforcement learning for optimal feedback control},
  author={Kamalapurkar, Rushikesh and Walters, Patrick and Rosenfeld, Joel and Dixon, Warren},
  year={2018},
  publisher={Springer}
}

@book{filippov2013differential,
  title={Differential equations with discontinuous righthand sides: control systems},
  author={Filippov, Aleksei Fedorovich},
  volume={18},
  year={2013},
  publisher={Springer Science \& Business Media}
}

@article{faghihian2025novel,
  title={A Novel Energy-Efficient Automated Regenerative Braking System},
journal={Applied Energy},
year={2025},
  author={Faghihian, Hamed and Sargolzei, Arman}
}

@article{faghihian2023energy,
  title={Energy Efficiency of Connected Autonomous Vehicles: A Review},
  author={Faghihian, Hamed and Sargolzaei, Arman},
  journal={Electronics},
  volume={12},
  number={19},
  pages={4086},
  year={2023},
  publisher={MDPI}
}

@incollection{FAGHIHIAN20241,
title = {Chapter 1 - Introduction to autonomous vehicles},
editor = {Muhammad H. Rashid},
booktitle = {Handbook of Power Electronics in Autonomous and Electric Vehicles},
publisher = {Academic Press},
pages = {1-16},
year = {2024},
isbn = {978-0-323-99545-0},
doi = {https://doi.org/10.1016/B978-0-323-99545-0.00018-X},
url = {https://www.sciencedirect.com/science/article/pii/B978032399545000018X},
author = {Hamed Faghihian and James Holland and Arman Sargolzaei},
}

@article{hung2021regionalized,
  title={Regionalized climate footprints of battery electric vehicles in Europe},
  author={Hung, Christine Roxanne and V{\"o}ller, Steve and Agez, Maxime and Majeau-Bettez, Guillaume and Str{\o}mman, Anders Hammer},
  journal={Journal of Cleaner Production},
  volume={322},
  pages={129052},
  year={2021},
  publisher={Elsevier}
}

@inproceedings{al2018minimizing,
  title={Minimizing energy consumption from connected signalized intersections by reinforcement learning},
  author={Al Islam, SMA Bin and Aziz, HM Abdul and Wang, Hong and Young, Stanley E},
  booktitle={2018 21st International Conference on Intelligent Transportation Systems (ITSC)},
  pages={1870--1875},
  year={2018},
  organization={IEEE}
}

@article{chen2024achieving,
  title={Achieving Energy-Efficient and Travel Time-Optimized Trajectory and Signal Control for CAEVs},
  author={Chen, Huiyu and Wu, Fan and Qiu, Tony Z},
  journal={IEEE Transactions on Intelligent Transportation Systems},
  year={2024},
  publisher={IEEE}
}

@article{mahmoud2021autonomous,
  title={Autonomous eco-driving with traffic light and lead vehicle constraints: An application of best constrained interpolation},
  author={Mahmoud, Yara Hazem and Brown, Nicholas E and Motallebiaraghi, Farhang and Koelling, Melinda and Meyer, Richard and Asher, Zachary D and Dontchev, Assen and Kolmanovsky, Ilya},
  journal={IFAC-PapersOnLine},
  volume={54},
  number={10},
  pages={45--50},
  year={2021},
  publisher={Elsevier}
}

@article{shang2022regenerative,
  title={Regenerative braking control strategy based on multi-source information fusion under environment perception},
  author={Shang, Yue and Ma, Chao and Yang, Kun and Tan, Di},
  journal={International Journal of Automotive Technology},
  volume={23},
  number={3},
  pages={805--815},
  year={2022},
  publisher={Springer}
}

@article{kim2021parameterized,
  title={Parameterized energy-optimal regenerative braking strategy for connected and autonomous electrified vehicles: A real-time dynamic programming approach},
  author={Kim, Dohee and Eo, Jeong Soo and Kim, Kwang-Ki K},
  journal={IEEE Access},
  volume={9},
  pages={103167--103183},
  year={2021},
  publisher={IEEE}
}

@incollection{faghihian2024introduction,
  title={Introduction to autonomous vehicles},
  author={Faghihian, Hamed and Holland, James and Sargolzaei, Arman},
  booktitle={Handbook of Power Electronics in Autonomous and Electric Vehicles},
  pages={1--16},
  year={2024},
  publisher={Elsevier}
}

@article{yao2021adaptive,
  title={Adaptive real-time optimal control for energy management strategy of extended range electric vehicle},
  author={Yao, Mingyao and Zhu, Bo and Zhang, Nong},
  journal={Energy Conversion and Management},
  volume={234},
  pages={113874},
  year={2021},
  publisher={Elsevier}
}

@article{naeem2024energy,
  title={Energy Efficient Solution for Connected Electric Vehicle and Battery Health Management Using Eco-driving Under Uncertain Environmental Conditions},
  author={Naeem, Hafiz Muhammad Yasir and Bhatti, Aamer Iqbal and Butt, Yasir Awais and Ahmed, Qadeer and Bai, Xiaoshan},
  journal={IEEE Transactions on Intelligent Vehicles},
  year={2024},
  publisher={IEEE}
}

@article{jeong2024adaptive,
  title={Adaptive robust electric vehicle routing under energy consumption uncertainty},
  author={Jeong, Jaehee and Ghaddar, Bissan and Zufferey, Nicolas and Nathwani, Jatin},
  journal={Transportation Research Part C: Emerging Technologies},
  volume={160},
  pages={104529},
  year={2024},
  publisher={Elsevier}
}

@article{lee2021energy,
  title={Energy management strategy of fuel cell electric vehicles using model-based reinforcement learning with data-driven model update},
  author={Lee, Heeyun and Cha, Suk Won},
  journal={IEEE Access},
  volume={9},
  pages={59244--59254},
  year={2021},
  publisher={IEEE}
}

@article{chen2020deep,
  title={Deep reinforcement learning based left-turn connected and automated vehicle control at signalized intersection in vehicle-to-infrastructure environment},
  author={Chen, Juan and Xue, Zhengxuan and Fan, Daiqian},
  journal={Information},
  volume={11},
  number={2},
  pages={77},
  year={2020},
  publisher={MDPI}
}

@inproceedings{isele2018navigating,
  title={Navigating occluded intersections with autonomous vehicles using deep reinforcement learning},
  author={Isele, David and Rahimi, Reza and Cosgun, Akansel and Subramanian, Kaushik and Fujimura, Kikuo},
  booktitle={2018 IEEE international conference on robotics and automation (ICRA)},
  pages={2034--2039},
  year={2018},
  organization={IEEE}
}

@article{guo2009generalized,
  title={Generalized Lyapunov method for discontinuous systems},
  author={Guo, Zhenyuan and Huang, Lihong},
  journal={Nonlinear Analysis: Theory, Methods \& Applications},
  volume={71},
  number={7-8},
  pages={3083--3092},
  year={2009},
  publisher={Elsevier}
}

@article{xiong2018battery,
  title={Battery and ultracapacitor in-the-loop approach to validate a real-time power management method for an all-climate electric vehicle},
  author={Xiong, Rui and Duan, Yanzhou and Cao, Jiayi and Yu, Quanqing},
  journal={Applied energy},
  volume={217},
  pages={153--165},
  year={2018},
  publisher={Elsevier}
}

@article{nyong2020reinforcement,
  title={Reinforcement learning based adaptive power pinch analysis for energy management of stand-alone hybrid energy storage systems considering uncertainty},
  author={Nyong-Bassey, Bassey Etim and Giaouris, Damian and Patsios, Charalampos and Papadopoulou, Simira and Papadopoulos, Athanasios I and Walker, Sara and Voutetakis, Spyros and Seferlis, Panos and Gadoue, Shady},
  journal={Energy},
  volume={193},
  pages={116622},
  year={2020},
  publisher={Elsevier}
}

@article{hu2019reinforcement,
  title={Reinforcement learning for hybrid and plug-in hybrid electric vehicle energy management: Recent advances and prospects},
  author={Hu, Xiaosong and Liu, Teng and Qi, Xuewei and Barth, Matthew},
  journal={IEEE Industrial Electronics Magazine},
  volume={13},
  number={3},
  pages={16--25},
  year={2019},
  publisher={IEEE}
}

@article{he2021improved,
  title={An improved energy management strategy for hybrid electric vehicles integrating multistates of vehicle-traffic information},
  author={He, Hongwen and Wang, Yunlong and Li, Jianwei and Dou, Jingwei and Lian, Renzong and Li, Yuecheng},
  journal={IEEE Transactions on Transportation Electrification},
  volume={7},
  number={3},
  pages={1161--1172},
  year={2021},
  publisher={IEEE}
}

@article{du2021heuristic,
  title={Heuristic energy management strategy of hybrid electric vehicle based on deep reinforcement learning with accelerated gradient optimization},
  author={Du, Guodong and Zou, Yuan and Zhang, Xudong and Guo, Lingxiong and Guo, Ningyuan},
  journal={IEEE Transactions on Transportation Electrification},
  volume={7},
  number={4},
  pages={2194--2208},
  year={2021},
  publisher={IEEE}
}

@inproceedings{haarnoja2018soft,
  title={Soft actor-critic: Off-policy maximum entropy deep reinforcement learning with a stochastic actor},
  author={Haarnoja, Tuomas and Zhou, Aurick and Abbeel, Pieter and Levine, Sergey},
  booktitle={International conference on machine learning},
  pages={1861--1870},
  year={2018},
  organization={PMLR}
}
\bibliographystyle{ieeetr}
\end{document}